\documentclass[conference,letterpaper]{IEEEtran}
\IEEEoverridecommandlockouts
\usepackage{cite}
\usepackage{url}
\usepackage{ifthen}
\usepackage{algorithm, algorithmicx, algpseudocode}
\usepackage{graphicx} 
\usepackage{textcomp}
\usepackage{booktabs}
\usepackage{makecell} 
\usepackage[cmex10]{amsmath}
\usepackage{amssymb,amsfonts}
\usepackage[left=1.62cm,right=1.62cm,top=1.85cm]{geometry}
\usepackage{pifont}
\usepackage{amsthm}
\usepackage{mathrsfs}
\def\BibTeX{{\rm B\kern-.05em{\sc i\kern-.025em b}\kern-.08em T\kern-.1667em\lower.7ex\hbox{E}\kern-.125emX}}
\usepackage{tikz}
\usetikzlibrary{automata,positioning,calc}
\usetikzlibrary{intersections}
\usetikzlibrary{decorations.pathreplacing,angles,quotes}

\usepackage{bm}
\usepackage{amscd}

\algnewcommand{\Initialize}[1]{%
  \State \textbf{Initialization:}
  \Statex \hspace*{\algorithmicindent}\parbox[t]{0.8\linewidth}{\raggedright #1}
}
\makeatletter

\theoremstyle{definition}
\newtheorem{theorem}{Theorem}
\newtheorem{definition}{Definition}
\newtheorem{lemma}{Lemma}
\newtheorem{prop}{Proposition}

\newtheorem{remark}{Remark}
\newtheorem{eg}{Example}

\newcommand{\norm}[1]{\left\lVert#1\right\rVert}

\newcommand{\relmiddle}[1]{\mathrel{}\middle#1\mathrel{}} 

\def\br{\mathbb R}

\def\vE{\mathbb E}

\font\b=cmr10 scaled\magstep4

\def\bigzerou{\smash{\lower1.7ex\hbox{\b 0}}}
\def\bigzerou{\smash{\lower1.7ex\hbox{\b 0}}}

\begin{document}
\title{Kolmogorov--Nagumo Mean Frameworks for Conditional Entropy
\thanks{This work was supported by JSPS KAKENHI Grant Numbers 26K14712 and JST CRONOS Grant Number JPMJCS25N5.}
}

\author{
\IEEEauthorblockN{Akira Kamatsuka}
\IEEEauthorblockA{Shonan Institute of Technology \\ 
Email: \text{kamatsuka@info.shonan-it.ac.jp}
 }
\and
\IEEEauthorblockN{Takahiro Yoshida}
\IEEEauthorblockA{Nihon University \\ 
Email: \text{yoshida.takahiro@nihon-u.ac.jp}
 } 
}
\maketitle

\begin{abstract}
This study focuses on conditional entropy frameworks based on the Kolmogorov--Nagumo (KN) mean.
First, $(\eta, \psi)$-KN averaging (\texttt{EPKNAVG}), a KN-mean extension of the $\eta$-averaging (\texttt{EAVG}) framework for $(\eta, F)$-entropies, is introduced and proven to be equivalent to \texttt{EAVG} under suitable concavification conditions.
Second, motivated by generalized $g$-vulnerability, a new framework is proposed for generalized $g$-conditional entropies.
This framework captures conditional entropies beyond the scope of \texttt{EAVG}-type representations.
In particular, it is shown that there exists an $\alpha$ and a joint probability distribution $p_{X, Y}$ such that the Augustin--Csisz{\' a}r conditional entropy $H_{\alpha}^{\mathrm{C}}(X|Y)$ cannot be represented by any $(\eta,F)$-entropy satisfying \texttt{EAVG}. In contrast, it is represented within the proposed framework.
Furthermore, sufficient conditions are derived under which the proposed generalized $g$-conditional entropies satisfy the conditioning reduces entropy property and the data-processing inequality.
\end{abstract}

\section{Introduction} \label{sec:intro}
Shannon entropy $H(X)$ and conditional entropy $H(X|Y)$ are among the most fundamental quantities in information theory.
Since their introduction, numerous generalized entropy measures have been proposed to capture diverse operational and statistical aspects of uncertainty.
Representative examples include R{\' e}nyi entropy $H_{\alpha}(X)$ \cite{renyi1961measures}, Havrda--Charv{\' a}t--Tsallis entropy $S_{\alpha}(X)$ \cite{Havrda1967QuantificationMO,Tsallis:1988aa}, and Sharma--Mittal entropy $H_{\alpha,\beta}(X)$ \cite{sharma1975new}.

Among the mathematical tools underlying these generalizations, generalized means play an important role.
In particular, the class now commonly referred to as the Kolmogorov--Nagumo (KN) mean, or quasi-arithmetic mean, originates from the classical works of Kolmogorov \cite{Kolmogorov1930}, Nagumo \cite{193071}, and de Finetti \cite{deFinetti1931}.
These generalized means were also studied from an axiomatic viewpoint (see, e.g., Hardy, Littlewood, and P{\' o}lya \cite{hardy1988inequalities}).
In information theory, R{\' e}nyi \cite{renyi1961measures} considered replacing the arithmetic mean in the mean-value postulate for Shannon entropy with a general KN mean.
Related axiomatic characterization results were subsequently developed by Dar{\' o}czy \cite{Daroczy:1964aa} (see also \cite{csiszar2011information,Solr-990001887550206443} for broader discussions of generalized means and axiomatic characterizations of information measures).

Along with generalized entropy measures, the notion of conditional entropy has also been extended in various directions.
For R{\' e}nyi entropy, several nonequivalent conditional versions have been introduced and extensively studied.
Examples include the Arimoto conditional entropy $H_{\alpha}^{\text{A}}(X|Y)$ \cite{arimoto1977} and the Hayashi conditional entropy $H_{\alpha}^{\text{H}}(X|Y)$ \cite{5773033}.
These conditional entropies satisfy important properties such as conditioning reduces entropy (\texttt{CRE}) and the data-processing inequality (\texttt{DPI}).
They also arise from distinct operational considerations.
The Arimoto conditional entropy is closely related to order-$\alpha$ channel capacity and Gallager-type error-exponent analysis \cite{arimoto1977}, whereas the Hayashi conditional entropy is used in the analysis
of leaked information in privacy amplification and related secrecy problems \cite{5773033}.
Extensive theoretical properties and operational interpretations of R{\' e}nyi-type conditional entropies have been reported in the literature \cite{6191351,10.1007/978-3-319-04268-8_7,6898022,8007073,ecea-4-05030,e22050526}.
Conditional versions of other generalized entropies have also been investigated.
In particular, several definitions of conditional entropy associated with the Havrda--Charv{\' a}t--Tsallis entropy have been studied in \cite{10.1063/1.2165744,Manije20133742,e23111427}.

Closely related research has emerged from statistical decision theory and quantitative information flow (QIF).
In statistical decision theory, Bayes loss \cite{grunwald2004} (also referred to as generalized entropy \cite{doi:10.1198/016214506000001437,Dawid:2007aa} or the Bayes envelope \cite{720534}) has been studied as a decision-theoretic measure of uncertainty induced by a decision problem.
In QIF, vulnerability, a dual notion to entropy, has been widely used to quantify information leakage.
For example, Alvim \textit{et al.} \cite{6266165,6957119,7536368,ALVIM201932} introduced $g$-vulnerability based on an adversarial threat model characterized by a gain function $g$.
Zarrabian and Sadeghi \cite{Zarrabian:2025aa} further generalized $g$-vulnerability by replacing the expectation operator with the KN mean.
This replacement allows the adversarial gain to be aggregated through a nonlinear mean rather than only through its expected value.
They also established a connection between generalized $g$-vulnerability and the Arimoto conditional entropy.
Recently, Kamatsuka and Yoshida \cite{11195284,11654062} derived similar connections for other R{\' e}nyi-type conditional entropies, including the Hayashi conditional entropy.
Additionally, they introduced the Augustin--Csisz{\' a}r conditional entropy $H_{\alpha}^{\text{C}}(X|Y)$ as a R{\' e}nyi-type conditional entropy by expressing the corresponding Augustin--Csisz{\' a}r mutual information \cite{augusting_phd_thesis,370121} as the difference between an entropy-like quantity and a conditional entropy-like quantity.

A different but related line of axiomatic research has focused on broad classes of entropy- and vulnerability-like measures rather than on the uniqueness characterization of a particular entropy functional.
Alvim \textit{et al.} \cite{ALVIM201932} and Am{\' e}rico \textit{et al.} \cite{9064819,9505206} developed general frameworks in which entropy and conditional entropy-like measures are studied through structural properties and aggregation rules.
Within these frameworks, entropy and conditional entropy-like measures are characterized through concavity (\texttt{CV}) together with aggregation operators such as averaging (\texttt{AVG}) and minimization (\texttt{MIN}).
They are also studied through core-concavity (\texttt{CCV}) combined with more general aggregation schemes, including $\eta$-averaging (\texttt{EAVG}) and the $\eta$-geometric mean (\texttt{EGM}).
Furthermore, Am{\' e}rico \textit{et al.} \cite{9505206} introduced a unifying class $\mathcal{Q}$ that encompasses a broad range of such constructions.
They also investigated conditions under which these conditional entropy-like measures satisfy \texttt{CRE} and \texttt{DPI}.
Recently, Zarrabian and Sadeghi \cite{Zarrabian:2025aa} studied analogous properties for generalized $g$-vulnerability.

The present study follows this latter line of research.
More specifically, we investigate how the use of a general KN mean affects the class of conditional entropy-like measures representable within these frameworks.
Our first question is whether replacing the expectation operator in the \texttt{EAVG} framework with a general KN mean enlarges the class of representable conditional entropies.
Our second question is whether generalized $g$-conditional entropy constructions based on KN means can represent conditional entropies that are not captured by \texttt{EAVG}.
Thus, our objective is not to derive a uniqueness characterization of a particular conditional entropy, but rather to clarify the scope and limitations of KN-mean-based extensions of existing frameworks for conditional entropy-like measures.

To address these questions, we study two KN-mean-based constructions.
The first is obtained by replacing the expectation operator in the \texttt{EAVG} framework with a general KN mean.
The second is based on generalized $g$-vulnerability and generalized $g$-Bayes vulnerability.
The main contributions of this study are summarized as follows:
\begin{itemize}
\item We introduce the $(\eta,\psi)$-KN averaging (\texttt{EPKNAVG}) class of conditional entropy, which extends the \texttt{EAVG} framework of Am{\' e}rico \textit{et al.} \cite{9064819,9505206} by incorporating the KN mean (Definition \ref{def:EPKNAVG}).
We prove that, under suitable concavification conditions, \texttt{EPKNAVG} is equivalent to \texttt{EAVG} (Theorem \ref{thm:EPKNAVG_EAVG_equivalent}).
Consequently, replacing the expectation operator with a general KN mean does not enlarge the class of representable conditional entropies under these conditions.
\item Building on the generalized $g$-vulnerability framework of Zarrabian and Sadeghi \cite{Zarrabian:2025aa} and the generalized $g$-Bayes vulnerability of \cite{11654062}, we study two generalized $g$-conditional entropy constructions.
The generalized $g$-posterior conditional entropy is directly related to the generalized posterior vulnerability of \cite{Zarrabian:2025aa}, whereas the generalized $g$-Bayes conditional entropy is based on the generalized $g$-Bayes vulnerability.
We prove that there exist an $\alpha\in(0,1)\cup(1,\infty)$ and a joint probability distribution $p_{X,Y}$ such that the Augustin--Csisz{\' a}r conditional entropy $H_{\alpha}^{\text{C}}(X|Y)$ cannot be represented within the \texttt{EAVG} framework of \cite{9064819,9505206}, whereas it admits a representation through the generalized $g$-Bayes construction (Proposition \ref{prop:AC_conditional_entropy_EAVG}).
We further establish \texttt{CRE} and sufficient conditions for \texttt{DPI} for the generalized $g$-conditional entropies (Theorems \ref{thm:CRE_DPI_g-posterior} and \ref{thm:CRE_DPI_g-Bayes}).
\end{itemize}

\section{Preliminaries}\label{sec:preliminaries}
Let $X$ and $Y$ be random variables taking values in finite alphabets $\mathcal{X}$ and $\mathcal{Y}$, respectively, with joint distribution $p_{X,Y} = p_{X} p_{Y|X}$.
Let $p_{Y}$ denote the marginal distribution of $Y$. 
Shannon and conditional entropies are defined by $H(X):=-\sum_{x}p_{X}(x)\log p_{X}(x)$, $H(X | Y):=-\sum_{x,y}p_{X,Y}(x,y)\log p_{X\mid Y}(x | y)$, respectively.
Let $\Delta_{\mathcal{X}}$ denote the probability simplex on $\mathcal{X}$. 
For a random variable $Z$ with distribution $p_{Z}$, the expectation and geometric mean are defined as $\mathbb{E}[Z]:=\sum_{z}p_{Z}(z)z$ and $\mathbb{G}[Z]:=\prod_{z}z^{p_{Z}(z)}$, respectively.
Throughout this paper, $\log$ denotes the natural logarithm.

In this section, we review several entropies and conditional entropies, along with axiomatic approaches to them.

\subsection{Various Entropies and Conditional Entropies} \label{ssec:entropy}
First, we begin by reviewing several generalized entropies and conditional entropies.

\begin{definition}
Let $\alpha\in (0, 1)\cup (1, \infty)$ and $\beta\neq 1$. For a probability distribution $p_{X}$, 
the \textit{R{\' e}nyi entropy of order $\alpha$} \cite{renyi1961measures}, the \textit{Havrda--Charv{\' a}t--Tsallis entropy of order $\alpha$} \cite{Havrda1967QuantificationMO,Tsallis:1988aa}, and the \textit{Sharma--Mittal entropy of order $(\alpha, \beta)$} \cite{sharma1975new} are defined as follows:
\begin{align}
H_{\alpha}(X)=H_{\alpha}(p_{X}) &:= \frac{1}{1-\alpha} \log \sum_{x} p_{X}(x)^{\alpha}, \label{eq:Renyi_ent} \\
S_{\alpha}(X)=S_{\alpha}(p_{X}) &:= \frac{1}{1-\alpha}\left( \sum_{x}p_{X}(x)^{\alpha} - 1 \right), \label{eq:Tsallis_ent} \\ 
H_{\alpha, \beta}(X)=H_{\alpha,\beta}(p_{X}) &:= \frac{1}{1-\beta} \left( \left( \sum_{x}p_{X}(x)^{\alpha} \right)^{\frac{1-\beta}{1-\alpha}} - 1 \right). \label{eq:SM_ent}
\end{align}
\end{definition}

\begin{remark}
We have $H_{\alpha, \alpha}(X) = S_{\alpha}(X)$, and each of the above entropies converges to the Shannon entropy in the limit as $\alpha=\beta\to 1$.
The entropies listed above are generally quasi-concave with respect to the distribution $p_{X}$, while some become concave under suitable conditions on the parameters. 
In particular, $H(X)$ and $S_{\alpha}(X)$ are concave in $p_{X}$ \cite[Thm. 2.7.3]{Cover:2006:EIT:1146355},\cite[Lemma 1]{10.1063/1.530920}.
Furthermore, for $0 < \alpha < 1$, the R{\' e}nyi entropy $H_{\alpha}(X)$ is concave \cite[Thm. 1]{1055890}. 
Additionally, the Sharma--Mittal entropy $H_{\alpha, \beta}(X)$ is concave whenever $\beta \geq 2-1/\alpha$, \cite[Prop. 29]{hoffmann2006concavity}.
\end{remark}

Notably, using the $\alpha$-norm $\norm{\cdot}_{\alpha}$ of the probability distribution $p_{X}$ together with the $q$-logarithm function $\ln_{q}t$, where $q\in\mathbb{R}$, defined below, all of the entropies can be represented in a unified form: 
\begin{align}
\ln_{q}t &:= 
\begin{cases}
\log t, & q=1, \\ 
\frac{t^{1-q}-1}{1-q}, & q\neq 1.
\end{cases} 
\end{align}

\begin{prop}
The following holds:
\begin{align}
H_{\alpha}(X) &= \log \norm{p_{X}}_{\alpha}^{\frac{\alpha}{1-\alpha}}, \label{eq:Renyi_ent_KN} \\
S_{\alpha}(X) &= \ln_{\alpha}\norm{p_{X}}_{\alpha}^{\frac{\alpha}{1-\alpha}}, \label{eq:Tsallis_ent_KN} \\ 
H_{\alpha, \beta}(X) &= \ln_{\beta}\norm{p_{X}}_{\alpha}^{\frac{\alpha}{1-\alpha}}. \label{eq:SM_ent_KN}
\end{align}
\end{prop}
\begin{proof}
These identities are derived through straightforward algebraic manipulations.
\end{proof}

For R{\' e}nyi entropy, several nonequivalent conditional versions have been proposed and extensively studied.
\begin{definition}
Let $\alpha\in (0, 1)\cup (1, \infty)$. Given a joint probability distribution $p_{X,Y}=p_{X}p_{Y\mid X}$, 
the \textit{Arimoto conditional entropy of order $\alpha$} \cite{arimoto1977}, the \textit{Hayashi conditional entropy of order $\alpha$} \cite{5773033}, and the \textit{Augustin--Csisz{\' a}r conditional entropy of order $\alpha$} \cite{11195284,11654062} are defined as follows: 
\begin{align}
H_{\alpha}^{\text{A}}(X \mid Y)&:= \frac{\alpha}{1-\alpha}\log\sum_{y} \left( \sum_{x}p_{X}(x)^{\alpha}p_{Y\mid X}(y\mid x)^{\alpha} \right)^{\frac{1}{\alpha}}, \label{eq:Arimoto_cond_renyi_ent} \\
H_{\alpha}^{\text{H}}(X \mid Y)&:= \frac{1}{1-\alpha}\log\sum_{y}p_{Y}(y) \sum_{x}p_{X\mid Y}(x \mid y)^{\alpha}, \label{eq:Hayashi_cond_renyi_ent} \\ 
H_{\alpha}^{\text{C}}(X\mid Y) &:= \min_{r_{X\mid Y}} \frac{\alpha}{1-\alpha}\notag \\ 
&\times \sum_{x}p_{X}(x)\log \sum_{y}p_{Y\mid X}(y\mid x)r_{X\mid Y}(x\mid y)^{1-\frac{1}{\alpha}}, \label{eq:AC_cond_renyi_ent}
\end{align}
where the minimum in \eqref{eq:AC_cond_renyi_ent} is taken over all conditional distributions $r_{X\mid Y}=\{r_{X\mid Y}(\cdot\mid y)\}_{y\in \mathcal{Y}}$.
\end{definition}
\begin{remark}
These R{\' e}nyi-type conditional entropies coincide with the conditional entropy $H(X | Y)$ in the limit as $\alpha \to 1$.
\end{remark}

Conditional counterparts of the Havrda--Charv{\' a}t--Tsallis entropy and the Sharma--Mittal entropy have been proposed in the literature.
\begin{definition}
Let $\alpha\in (0, 1)\cup (1, \infty)$ and $\beta\neq 1$. Given a joint probability distribution $p_{X,Y}=p_{X}p_{Y\mid X}$, 
the Havrda--Charv{\' a}t--Tsallis conditional entropy \cite{Manije20133742} and the Sharma--Mittal conditional entropy \cite{9064819,e24010039} are respectively defined as follows: 
\begin{align}
S_{\alpha}^{\text{M}}(X\mid Y) &:= \mathbb{E}_{Y}\left[S_{\alpha}(p_{X\mid Y}(\cdot \mid Y))\right], \label{eq:Manije_cond_Tsallis_ent} \\
H_{\alpha, \beta}^{\text{AKM}}(X\mid Y) &:= \ln_{\beta} \left( \mathbb{E}_{Y}\left[\norm{p_{X\mid Y}(\cdot\mid Y)}_{\alpha}^{\alpha}\right] \right)^{\frac{1}{1-\alpha}}. \label{eq:Americo_cond_SM_ent}
\end{align}
\end{definition}
\begin{remark}
Alternative conditional versions of the Havrda--Charv{\' a}t--Tsallis conditional entropy have been proposed by Furuichi \cite{10.1063/1.2165744} and Teixeira \textit{et al.} \cite{e23111427}.
\end{remark}

In QIF, Zarrabian and Sadeghi \cite{Zarrabian:2025aa}, as well as Kamatsuka and Yoshida \cite{11654062} introduced generalized $g$-vulnerability, a dual notion of entropy based on the KN mean. 
Suppose that an adversary, characterized by a gain function $g(x, a)$, selects an action $A$ concerning the original data $X$ after observing the disclosed information $Y$.
The adversary employs a decision rule $\delta\colon\mathcal{Y}\to\mathcal{A}$ together with a gain function $g\colon\mathcal{X}\times\mathcal{A}\to\mathbb{R}$, where $\mathcal{A}$ denotes the action space.

\begin{definition}[Kolmogorov--Nagumo (KN) mean]
Let $I\subseteq \mathbb{R}$ be an interval and $\varphi\colon I\to \mathbb{R}$ be a strictly monotonic and continuous function. 
Given a random variable $Z$ with distribution $p_{Z}$, the \textit{KN mean} (also known as the \textit{quasi-arithmetic mean}) of $Z$ is defined as follows:
\begin{align}
\mathbb{M}_{\varphi}[Z] &:= \varphi^{-1} \left( \mathbb{E}\left[\varphi(Z)\right] \right) \\
&= \varphi^{-1}\left( \sum_{z}p_{Z}(z)\varphi(z) \right),
\end{align}
where $\varphi^{-1}$ denotes the inverse function of $\varphi$. 
Similarly, given another random variable $Y$, the \textit{conditional KN mean} of $Z$ given $Y=y$ is defined by
\begin{align}
\mathbb{M}_{\varphi}[Z\mid Y=y]
&:= \varphi^{-1}\left(\mathbb{E}[\varphi(Z)\mid Y=y]\right) \\
&= \varphi^{-1}\left(\sum_{z}p_{Z\mid Y}(z\mid y)\varphi(z)\right).
\end{align}
\end{definition}

\begin{eg}
Examples of the KN mean include the following: 
\begin{itemize}
\item If $\varphi(t)=at + b$, $a\neq 0$, then $\mathbb{M}_{\varphi}[Z] = \mathbb{E}[Z]$.
\item If $\varphi(t)=\log t$, then $\mathbb{M}_{\varphi}[Z] = \mathbb{G}[Z]$. 
\item If $\varphi(t)=\ln_{q} t$, then $\mathbb{M}_{\varphi}[Z] = \mathbb{H}_{1-q}[Z] := \{\mathbb{E}[Z^{1-q}]\}^{\frac{1}{1-q}}$ (H{\" o}lder mean of order $1-q$). 
\end{itemize}
\end{eg}

\begin{definition}[Generalized vulnerability] \label{def:generalized_vulnerability}
Let $(X,Y)\sim p_{X}p_{Y|X}$ and $g(x,a)$ be a gain function.
Let $I\subseteq \mathbb{R}$ be an interval and $\varphi,\psi\colon I\to\mathbb{R}$ be continuous and strictly monotone functions.
The \textit{generalized $g$-prior vulnerability}, \textit{generalized average $g$-posterior conditional vulnerability} \cite{Zarrabian:2025aa}, and \textit{generalized $g$-Bayes conditional vulnerability} \cite{11654062} are defined as follows:
\begin{align}
V_{\varphi, g}(X) &= \max_{a} \mathbb{M}_{\varphi}[g(X, a)], \label{eq:generalized_g_prior} \\ 
\hat{V}_{\psi, \varphi, g}(X\mid Y) &= \mathbb{M}_{\psi}\left[\max_{a} \mathbb{M}_{\varphi}\left[g(X, a)\relmiddle{|} Y\right]\right], \label{eq:generalized_g_posterior} \\
V_{\varphi, \psi, g}(X\mid Y) &= \max_{\delta} \mathbb{M}_{\varphi}\left[\mathbb{M}_{\psi}\left[g(X, \delta(Y))\relmiddle{|}X\right]\right], \label{eq:generalized_g_Bayes}
\end{align}
where the maximum in Eq.~\eqref{eq:generalized_g_Bayes} is taken over all decision rules $\delta\colon \mathcal{Y}\to \mathcal{A}$. 
\end{definition}
\begin{remark} \label{rem:generalized_vulnerability}
While $\hat{V}_{\psi, \varphi, g}(X | Y)$ \cite{Zarrabian:2025aa} generalizes the optimal value of posterior expected gain, 
$V_{\varphi, \psi, g}(X | Y)$ \cite{11654062} generalizes the optimal value of Bayes expected gain via KN means. 
Generally, $V_{\varphi,\psi,g}(X | Y)$ and $\hat{V}_{\psi,\varphi,g}(X | Y)$ do not coincide. However, they agree when $\varphi=\psi$, i.e.,
$\hat{V}_{\varphi, \varphi, g}(X | Y) = V_{\varphi, \varphi, g}(X | Y)$, as shown in \cite[Prop. 3]{11654062}.
\end{remark}

The generalized $g$-vulnerabilities satisfy the following fundamental properties.
\begin{prop} \label{prop:g-vulnerability_basic_propery}
Let $\eta$ be a strictly increasing and continuous function. Then, 
\begin{align}
\eta \left( \hat{V}_{\psi, \varphi, g}(X\mid Y) \right) = \hat{V}_{\psi\circ\eta^{-1}, \varphi\circ \eta^{-1}, \eta\circ g}(X\mid Y), \\
\eta \left( V_{\varphi, \psi, g}(X\mid Y) \right) = V_{\varphi\circ \eta^{-1}, \psi\circ \eta^{-1}, \eta\circ g}(X\mid Y).
\end{align}
\end{prop}
\begin{proof}
\begin{align}
&\eta \left( \hat{V}_{\psi, \varphi, g}(X\mid Y) \right) \notag \\ 
&= \eta \left( \mathbb{M}_{\psi}\left[\max_{a}\mathbb{M}_{\varphi}[g(X, a)\mid Y]\right] \right) \\ 
&= \mathbb{M}_{\psi\circ \eta^{-1}}\left[\max_{a}\mathbb{M}_{\varphi\circ \eta^{-1}}[\eta\circ g(X, a)\mid Y]\right] \\ 
&= \hat{V}_{\psi\circ\eta^{-1}, \varphi\circ \eta^{-1}, \eta\circ g}(X\mid Y).
\end{align}
Similarly, 
\begin{align}
&\eta \left( V_{\varphi, \psi, g}(X\mid Y) \right) \notag \\ 
&= \eta\left(\max_{\delta}\mathbb{M}_{\varphi}[\mathbb{M}_{\psi}[g(X, \delta(Y))\mid X]] \right) \\ 
&= \max_{\delta} \eta\left(\mathbb{M}_{\varphi}[\mathbb{M}_{\psi}[g(X, \delta(Y))\mid X]] \right) \\ 
&= \max_{\delta} \mathbb{M}_{\varphi\circ \eta^{-1}}[\mathbb{M}_{\psi\circ \eta^{-1}}[\eta\circ g(X, \delta(Y))\mid X]] \\
&= V_{\varphi\circ\eta^{-1}, \psi\circ \eta^{-1}, \eta\circ g}(X\mid Y).
\end{align}
\end{proof}

\subsection{Axiomatic Approaches to Entropies and Conditional Entropies} \label{ssec:QIF}
Next, we review the axiomatic approach for entropies and conditional entropies by Alvim \textit{et al.} \cite{6266165,6957119,7536368,ALVIM201932} and Am{\' e}rico \textit{et al.} \cite{9064819,9505206}.  

\begin{definition}[$(\eta, F)$-entropy, \text{\cite[Def. 11]{9505206}}]
Given a probability distribution $p_{X}$, a continuous function $F\colon \Delta_{\mathcal{X}}\to \mathbb{R}$, and a strictly increasing and continuous function $\eta\colon F(\Delta_{\mathcal{X}})\to \mathbb{R}$, 
$\mathcal{H}(X)=\mathcal{H}(p_{X})$ is an \textit{$(\eta, F)$-entropy}, denoted by $\mathcal{H}=(\eta, F)$, if 
\begin{align}
\mathcal{H}(X) &= \eta(F(p_{X})).
\end{align}
\end{definition}

\begin{definition}[\texttt{CCV}, \text{\cite[Def. 12]{9505206}}]
An entropy $\mathcal{H}=(\eta, F)$ is said to be \textit{core-concave} (\texttt{CCV}) if the function $F$ is concave.
\end{definition}

\begin{definition}[\texttt{EAVG}, \text{\cite[Def. 13]{9505206}}]
Given a joint probability distribution $p_{X, Y}=p_{X} p_{Y\mid X}$, an entropy $\mathcal{H}=(\eta, F)$ is said to satisfy \textit{$\eta$-averaging} (\texttt{EAVG}) if the conditional entropy $\mathcal{H}(X | Y)=\mathcal{H}(p_{X}, p_{Y\mid X})$ is defined as follows: 
\begin{align}
\mathcal{H}(X\mid Y) &= \eta \left( \mathbb{E}_{Y}\left[F(p_{X\mid Y}(\cdot \mid Y))\right] \right) \\ 
&= \eta\left( \sum_{y}p_{Y}(y)F(p_{X\mid Y}(\cdot \mid y)) \right).
\end{align}
\end{definition}

\begin{definition}[\texttt{CRE}, \texttt{DPI}]
Given a joint probability distribution $p_{X, Y}=p_{X} p_{Y\mid X}$, $(\mathcal{H}(X), \mathcal{H}(X | Y))$ is said to satisfy the \texttt{CRE} if 
\begin{align}
\mathcal{H}(X\mid Y)\leq \mathcal{H}(X).
\end{align}
Furthermore, $\mathcal{H}(X | Y)$ is said to satisfy the \texttt{DPI} if, whenever $X-Y-Z$ forms a Markov chain, 
\begin{align}
\mathcal{H}(X\mid Y) &\leq \mathcal{H}(X\mid Z).
\end{align}
\end{definition}

\begin{prop}[\text{\cite[Thm. 2]{9064819}}] \label{prop:EAVG_CRE_DPI_CCV_equivalent}
If an entropy $\mathcal{H}=(\eta, F)$ satisfies \texttt{EAVG}, then the properties \texttt{CRE}, \texttt{DPI}, and \texttt{CCV} are equivalent.
\end{prop}

\begin{remark}
The entropy measures $H(X)$, $H_{\alpha}(X)$, $S_{\alpha}(X)$, $H_{\alpha, \beta}(X)$, and $-V_{\varphi, g}(X)$ satisfy the \texttt{CCV} condition under appropriate assumptions on $\varphi$ \cite[Thm. 11]{Zarrabian:2025aa}. 
Furthermore, the conditional entropies $H_{\alpha}^{\text{A}}(X | Y)$, $H_{\alpha}^{\text{H}}(X | Y)$, $S_{\alpha}^{\text{M}}(X | Y)$, $H_{\alpha, \beta}^{\text{AKM}}(X | Y)$, and $-\hat{V}_{\psi, \varphi, g}(X|Y)$ satisfy the \texttt{EAVG} condition under appropriate conditions on $\varphi, \psi$, and $g$ (see \cite[Table 1]{e24010039} and Section \ref{sec:generalized_g_conditional_entropy}). 
In contrast, $H_{\alpha}^{\text{C}}(X | Y)$ and $-V_{\varphi, \psi, g}(X | Y)$ do not satisfy the \texttt{EAVG} condition in general. 
Nevertheless, $(H(X), H_{\alpha}^{\text{C}}(X | Y))$ and $(-V_{\varphi, g}(X), -V_{\varphi, \psi, g}(X | Y))$ satisfy the properties \texttt{CRE} and \texttt{DPI} under suitable conditions on $\varphi$, $\psi$, and $g$ (see \cite[Prop. 5]{11195284} and Section \ref{sec:generalized_g_conditional_entropy}).
\end{remark}

\begin{definition}[\texttt{EGM}, \text{\cite[Def. 15]{9505206}}]
Given a joint probability distribution $p_{X, Y}=p_{X} p_{Y\mid X}$, an entropy $\mathcal{H}=(\eta, F)$ is said to satisfy \textit{$\eta$-geometric mean} (\texttt{EGM}) if the conditional entropy $\mathcal{H}(X | Y)=\mathcal{H}(p_{X}, p_{Y\mid X})$ is defined as follows: 
\begin{align}
\mathcal{H}(X\mid Y) &:= \eta \left( \mathbb{G}\left[F(p_{X\mid Y}(\cdot \mid Y))\right] \right) \\ 
&= \eta\left( \prod_{y}F(p_{X\mid Y}(\cdot \mid y))^{p_{Y}(y)} \right).
\end{align}
\end{definition}
\begin{prop}[\text{\cite[Thm. 8]{9505206}}]
If an entropy $\mathcal{H}=(\eta, F)$ satisfies \texttt{CCV} and \texttt{EGM}, and if $F$ is non-negative, then the pair $(\mathcal{H}(X), \mathcal{H}(X|Y))$ satisfies the properties \texttt{CRE} and \texttt{DPI}.
\end{prop}
\begin{remark}
Alvim \textit{et al.} \cite{ALVIM201932} proposed another set of axioms based on quasiconcavity (\texttt{QCV}) and minimum (\texttt{MIN}) \footnote{Because their original axiomatization was formulated for vulnerability, which is dual to entropy, they employed quasiconvexity (\texttt{QCVX}) and maximum (\texttt{MAX}) instead of quasiconcavity and minimum.}. 
Subsequently, Am{\' e}rico \textit{et al.} \cite{9505206} introduced a unifying entropy class $\mathcal{Q}$ that encompasses all the constructions described above.
\end{remark}

\section{Kolmogorov--Nagumo Mean Framework for Conditional Entropy}\label{sec:KN_framework}

In this section, we propose a new framework for conditional entropy, referred to as $(\eta, \psi)$-KN averaging (\texttt{EPKNAVG}), by extending the framework of Am{\' e}rico \textit{et al.} \cite{9064819,9505206} through the replacement of the expectation operator and geometric mean with the KN mean.

\begin{definition}[\texttt{EPKNAVG}] \label{def:EPKNAVG}
Let $I\subseteq \br$ be an interval such that $F(\Delta_{\mathcal{X}})\subseteq I$, and let $\psi\colon I\to\mathbb{R}$ be a strictly monotonic and continuous function.
An entropy $\mathcal{H}=(\eta, F)$ is said to satisfy \textit{$(\eta, \psi)$-KN averaging} (\texttt{EPKNAVG}) if the conditional entropy $\mathcal{H}(X | Y)$ is defined by: 
\begin{align}
\mathcal{H}(X \mid Y) = \eta\left( \mathbb{M}_{\psi}[F(p_{X\mid Y}(\cdot\mid Y))] \right).
\end{align}
\end{definition}

\begin{theorem} \label{thm:EPKNAVG_EAVG_equivalent}
Assume that either
\begin{itemize}
\item $\psi$ is strictly increasing and $\psi\circ F$ is concave\footnote{In utility theory, this condition is referred to as the \textit{concavification} of $F$ by $\psi$ (see, e.g., \cite{RePEc:iuk:wpaper:2012-10}).}, or
\item $\psi$ is strictly decreasing and $\psi\circ F$ is convex.
\end{itemize}
Then, an entropy $\mathcal{H}$ satisfies \texttt{EPKNAVG} for some $(\eta,F,\psi)$ if and only if it satisfies \texttt{EAVG} for some $(\tilde{\eta},\tilde{F})$.
\end{theorem}

\begin{proof}
See Appendix \ref{proof:EPKNAVG_EAVG_equivalent}. 
\end{proof}

\begin{remark}
Theorem \ref{thm:EPKNAVG_EAVG_equivalent} shows that, under the stated assumptions, replacing the expectation operator in the \texttt{EAVG} framework with the KN mean does not enlarge the class of representable conditional entropies.
\end{remark}

\section{Generalized $g$-Conditional Entropies} \label{sec:generalized_g_conditional_entropy}

In this section, we propose a new framework for conditional entropy based on the generalized $g$-vulnerability introduced in Definition \ref{def:generalized_vulnerability}. 
We show that the Augustin--Csisz{\' a}r conditional entropy defined in Eq.~\eqref{eq:AC_cond_renyi_ent} belongs to this framework, whereas Proposition \ref{prop:AC_conditional_entropy_EAVG}
shows that an \texttt{EAVG} representation does not exist in general.
Furthermore, we derive sufficient conditions under which the entropy measures defined within the proposed framework satisfy the properties \texttt {CRE} and \texttt{DPI}.

\begin{definition}[($\varphi, g$)-entropy]
Let $I\subseteq\mathbb{R}$ be an interval, $\varphi\colon I\to\mathbb{R}$ be a strictly monotonic and continuous function,
and $g(x, a)$ be a gain function.
Given a probability distribution $p_{X}$, the \textit{$(\varphi, g)$-entropy} is defined by: 
\begin{align}
H_{\varphi, g}(X) &= -V_{\varphi, g}(X).
\end{align}
\end{definition}

\begin{definition}[Generalized $g$-conditional entropies]
Let $I\subseteq\mathbb{R}$ be an interval, $\varphi,\psi\colon I\to\mathbb{R}$ be strictly monotonic and continuous functions, 
and $g(x, a)$ be a gain function.
Given a joint probability distribution $p_{X, Y}=p_{X}p_{Y\mid X}$, the \textit{$(\psi, \varphi, g)$-posterior conditional entropy} and the \textit{$(\varphi, \psi, g)$-Bayes conditional entropy} are respectively defined as follows: 
\begin{align}
\hat{H}_{\psi, \varphi, g}(X|Y) &:= -\hat{V}_{\psi, \varphi, g}(X|Y), \\
H_{\varphi, \psi, g}(X|Y) &:= -V_{\varphi, \psi, g}(X|Y).
\end{align}
\end{definition}

\begin{remark}
By Remark \ref{rem:generalized_vulnerability}, $\hat{H}_{\psi, \varphi, g}(X | Y)$ and $H_{\varphi, \psi, g}(X | Y)$ do not generally coincide, although they agree when $\varphi=\psi$.
\end{remark}

\begin{eg}
Let $\mathcal{A}=\Delta_{\mathcal{X}}$ and $g_{\text{$0$-$1$}}$ denote the soft $0$-$1$ score defined by $g_{\text{$0$-$1$}}(x, r) := r(x)$. 
By \cite[Prop. 2]{Zarrabian:2025aa}, \cite[Thm. 1]{11654062}, and Proposition \ref{prop:g-vulnerability_basic_propery}, 
Shannon entropy and several conditional entropies, including R{\'e}nyi-type ones, can be represented within the framework of generalized $g$-conditional entropies as follows:
\begin{align}
H(X) &= -\log V_{\log t,  g_{\text{0-1}}}(X) \\ 
&=-V_{t, \log g_{\text{$0$-$1$}}}(X), \\ 
H(X|Y) &= -\log V_{\log t, \log t, g_{\text{0-1}}}(X|Y) \\ 
&=-\hat{V}_{t, t, \log g_{\text{$0$-$1$}}}(X|Y), \\ 
&=-V_{t, t, \log g_{\text{$0$-$1$}}}(X|Y), \\ 
H_{\alpha}^{\text{A}}(X|Y) &= -\log V_{\ln_{\frac{1}{\alpha}}t, \ln_{\frac{1}{\alpha}}t, g_{\text{$0$-$1$}}}(X|Y) \\
&= -\hat{V}_{\ln_{\frac{1}{\alpha}}\circ \exp\{t\}, \ln_{\frac{1}{\alpha}}\circ \exp\{t\}, \log g_{\text{$0$-$1$}}}(X|Y) \\ 
&= -V_{\ln_{\frac{1}{\alpha}}\circ \exp\{t\}, \ln_{\frac{1}{\alpha}}\circ \exp\{t\}, \log g_{\text{$0$-$1$}}}(X|Y), \\ 
H_{\alpha}^{\text{C}}(X|Y)&= -\log V_{\log t, \ln_{\frac{1}{\alpha}}t, g_{\text{$0$-$1$}}}(X|Y) \\
&= -V_{t, \ln_{\frac{1}{\alpha}}\circ \exp\{t\}, \log g_{\text{$0$-$1$}}}(X|Y).
\end{align}
\end{eg}

Next, we establish an impossibility result that there exists an $\alpha$ and $p_{X, Y}$ for which the Augustin--Csisz{\' a}r conditional entropy $H_{\alpha}^{\text{C}}(X | Y)$ cannot be represented by any $(\eta, F)$-entropy satisfying the \texttt{EAVG} condition. 

\begin{prop} \label{prop:AC_conditional_entropy_EAVG}
There exists an $\alpha\in (0, 1)\cup (1, \infty)$ and a joint probability distribution $p_{X, Y}$ such that $H_{\alpha}^{\text{C}}(X|Y)$ cannot be represented by any $(\eta, F)$-entropy satisfying \texttt{EAVG}.
\end{prop}

\begin{proof}
See Appendix \ref{proof:AC_conditional_entropy_EAVG}. 
\end{proof}

\begin{remark}
Although Proposition \ref{prop:AC_conditional_entropy_EAVG} shows that $H_{\alpha}^{\text{C}}(X|Y)$ cannot be represented within the \texttt{EAVG} framework for some $\alpha$ and $p_{X,Y}$, it is known to satisfy \texttt{CRE} and \texttt{DPI} \cite[Prop. 5]{11195284}.
\end{remark}

In the following, we derive sufficient conditions for generalized $g$-conditional entropies. 

\begin{lemma} \label{lem:generalized_g_vulnarability_property}
Given a joint probability distribution
$p_{X,Y}=p_{X}p_{Y\mid X}$, let $\eta(t):=t$ and define $F(p_{X}):=-V_{\varphi,g}(X)$.
Furthermore, define the reflected KN generator $\bar{\psi}(t):=\psi(-t)$.
If $\varphi$ is strictly concave, then the following hold:
\begin{enumerate}
\item
The $(\varphi,g)$-entropy $H_{\varphi,g}(X)$ is an
$(\eta,F)$-entropy satisfying \texttt{CCV}.
\item
The $(\psi,\varphi,g)$-posterior conditional entropy
$\hat{H}_{\psi,\varphi,g}(X|Y)$ satisfies
\texttt{EPKNAVG} with respect to the $(\eta,F)$-entropy
and the KN generator $\bar{\psi}$.
\end{enumerate}
\end{lemma}

\begin{proof}
$1)$ follows immediately from \cite[Thm. 11]{Zarrabian:2025aa}. 
For $2)$, since $\bar{\psi}^{-1}(t)=-\psi^{-1}(t)$, we obtain
$\mathbb{M}_{\bar{\psi}}\left[F(p_{X\mid Y}(\cdot\mid Y))\right]=\hat{H}_{\psi,\varphi,g}(X|Y)$.
Therefore,
$\hat{H}_{\psi,\varphi,g}(X|Y)$ satisfies
\texttt{EPKNAVG} with KN generator $\bar{\psi}$.
\end{proof}

\begin{theorem} \label{thm:CRE_DPI_g-posterior}
Given a joint probability distribution
$p_{X,Y}=p_{X}p_{Y\mid X}$, let $F(p_{X}):=-V_{\varphi,g}(X),\bar{\psi}(t):=\psi(-t)$.
Assume that $\varphi$ is strictly concave and that either
\begin{itemize}
\item
$\bar{\psi}$ is strictly increasing and
$\bar{\psi}\circ F$ is concave, or
\item
$\bar{\psi}$ is strictly decreasing and
$\bar{\psi}\circ F$ is convex.
\end{itemize}
Then,
$\hat{H}_{\psi,\varphi,g}(X|Y)$ satisfies \texttt{DPI}.
Consequently,
$(H_{\varphi,g}(X),
\hat{H}_{\psi,\varphi,g}(X|Y))$
satisfies \texttt{CRE}.
\end{theorem}

\begin{proof}
By Lemma \ref{lem:generalized_g_vulnarability_property},
$H_{\varphi,g}(X)$ is an $(\eta,F)$-entropy satisfying
\texttt{CCV}, and
$\hat{H}_{\psi,\varphi,g}(X|Y)$ satisfies
\texttt{EPKNAVG} with KN generator $\bar{\psi}$.
Under the stated assumptions,
Theorem \ref{thm:EPKNAVG_EAVG_equivalent}
yields an equivalent \texttt{EAVG} representation satisfying
\texttt{CCV}.
Hence, Proposition \ref{prop:EAVG_CRE_DPI_CCV_equivalent} implies \texttt{DPI}. The \texttt{CRE} property follows as a special case of \texttt{DPI} by taking the downstream variable $Z$ to be constant.
\end{proof}

\begin{remark}
Zarrabian and Sadeghi \cite[Prop. 8]{Zarrabian:2025aa} established a related sufficient condition for \texttt{DPI}.
\end{remark}

\begin{theorem} \label{thm:CRE_DPI_g-Bayes}
Given a joint probability distribution $p_{X,Y}=p_{X}p_{Y\mid X}$, the following holds:
\begin{itemize}
\item $(H_{\varphi,g}(X),H_{\varphi,\psi,g}(X|Y))$ satisfies \texttt{CRE}.
\item Suppose that $\mathcal{A}$ is a convex set and that either
\begin{itemize}
\item $\psi$ is strictly increasing and $\psi\circ g$ is concave in $a$, or
\item $\psi$ is strictly decreasing and $\psi\circ g$ is convex in $a$.
\end{itemize}
Then, $H_{\varphi,\psi,g}(X|Y)$ satisfies \texttt{DPI}.
\end{itemize}
\end{theorem}

\begin{proof}
See Appendix \ref{proof:CRE_DPI_g-Bayes}.
\end{proof}

\begin{remark}
The \texttt{DPI} property implies \texttt{CRE} by taking the downstream variable $Z$ to be constant. 
In Theorem \ref{thm:CRE_DPI_g-Bayes}, however, \texttt{CRE} holds under weaker assumptions than those required for \texttt{DPI}.
In particular, the convexity of $\mathcal{A}$ and the concavity or convexity condition on $\psi\circ g$ are required only for \texttt{DPI}.
\end{remark}

\section{Conclusion}\label{sec:conclusion}

In this study, we investigated conditional entropy frameworks that are based on the KN mean.
First, we introduced the \texttt{EPKNAVG} framework as a KN-mean extension of the \texttt{EAVG} framework.
Under suitable conditions, we proved that \texttt{EPKNAVG} is equivalent to \texttt{EAVG}.
Next, we proposed a new framework of generalized $g$-conditional entropies based on generalized $g$-vulnerability.
This framework encompasses conditional entropies that lie beyond the scope of \texttt{EAVG}-type representations.
In particular, we showed that there exists an $\alpha$ and a joint probability distribution $p_{X, Y}$ such that the Augustin--Csisz{\' a}r conditional entropy $H_{\alpha}^{\mathrm{C}}(X|Y)$ cannot be represented by any $(\eta,F)$-entropy satisfying $\texttt{EAVG}$, whereas it admits a natural representation within the proposed framework.
Finally, we derived sufficient conditions under which the proposed generalized $g$-conditional entropies satisfy the properties \texttt{CRE} and \texttt{DPI}.

An important direction for future work is to characterize the relationship between generalized $g$-conditional entropies and the full class $\mathcal{Q}$, and to investigate further axiomatic properties of KN-mean-based conditional entropy frameworks. Another direction is to identify operational problems in which the generalized $g$-Bayes construction arises naturally.

\appendices

\section{Proof of Theorem \ref{thm:EPKNAVG_EAVG_equivalent}} \label{proof:EPKNAVG_EAVG_equivalent}
\begin{proof}
$(\texttt{EAVG} \Rightarrow \texttt{EPKNAVG})$: 
Assume that $\mathcal{H}$ satisfies \texttt{EAVG} for some $(\tilde{\eta}, \tilde{F})$.
Then, by setting $\psi(t)=t$ and $F=\tilde{F}$, $\mathcal{H}=(\tilde{\eta}, F)$ satisfies \texttt{EPKNAVG}.

$(\texttt{EPKNAVG} \Rightarrow \texttt{EAVG})$: 
Assume that $\mathcal{H}$ satisfies \texttt{EPKNAVG} for some $(\eta, F)$. 
Then, 
\begin{align}
\mathcal{H}(X\mid Y) &= \eta\left( \mathbb{M}_{\psi}\left[F(p_{X\mid Y}(\cdot\mid Y))\right] \right) \\
&= \eta \left( \psi^{-1}\left( \sum_{y}p_{Y}(y)\psi(F(p_{X\mid Y}(\cdot\mid y))) \right) \right) \\ 
&= \eta\circ \psi^{-1} \left( \vE_{Y}\left[\psi\circ F(p_{X\mid Y}(\cdot\mid Y))\right] \right).
\end{align}
Thus, if $\psi$ is strictly increasing and $\psi\circ F$ is concave, then $\mathcal{H}=(\eta\circ \psi^{-1}(t), \psi\circ F(t))$ satisfies \texttt{EAVG}, because $\eta\circ \psi^{-1}$ is strictly increasing. 
Similarly, if $\psi$ is strictly decreasing and $\psi\circ F$ is convex, then $-\psi\circ F$ is concave and $\eta\circ\psi^{-1}(-t)$ is strictly increasing. Hence, $\mathcal{H}=(\eta\circ
\psi^{-1}(-t), -\psi\circ F(t))$ satisfies \texttt{EAVG}.
\end{proof}

\section{Proof of Proposition \ref{prop:AC_conditional_entropy_EAVG}} \label{proof:AC_conditional_entropy_EAVG}
\begin{proof}
Let $\alpha=2$, $\mathcal{X}=\mathcal{Y}=\{0, 1\}$, and consider two distributions $p_{X}^{(i)}$, $i=0, 1$, with
\[
(p_{X}^{(0)}(0), p_{X}^{(0)}(1))=(0.9, 0.1), \quad (p_{X}^{(1)}(0), p_{X}^{(1)}(1))=(0.1, 0.9).
\]
Suppose, for contradiction, that there exists an $(\eta, F)$-entropy satisfying \texttt{EAVG} such that 
\begin{align}
H_{2}^{\text{C}}(X|Y) &= \eta\left( \vE_{Y}\left[F(p_{X\mid Y}(\cdot \mid Y))\right] \right).
\end{align}
First, consider the case in which $Y$ is independent of $X$ and $p_{X\mid Y}=p_{X}^{(i)}$ for $i=0,1$. Using \cite[Eq.~(23)]{370121}, we obtain
\begin{align}
H_{2}^{\text{C}}(X|Y) &= H(p_{X}^{(i)}) = \eta(F(p_{X}^{(i)})), \qquad i=0,1.
\end{align}
By the symmetry of Shannon entropy, 
\begin{align}
H(p_{X}^{(0)}) = -0.9\log 0.9 - 0.1\log 0.1 = H(p_{X}^{(1)}).
\end{align}
Thus, we obtain $\eta(F(p_{X}^{(0)})) = \eta(F(p_{X}^{(1)}))$, which implies that
\begin{align}
F(p_{X}^{(0)}) = F(p_{X}^{(1)}). \label{eq:identity}
\end{align} 
Now, consider the following joint probability distribution:
\[
p_{X, Y}(x, i)=\frac{1}{2}p_{X}^{(i)}(x), \qquad i=0,1.
\]
Then, 
\begin{align}
H_{2}^{\text{C}}(X|Y) &= \eta\left( \vE_{Y}\left[F(p_{X\mid Y}(\cdot\mid Y))\right] \right) \\ 
&= \eta\left( \frac{1}{2} F(p_{X}^{(0)}) + \frac{1}{2}F(p_{X}^{(1)}) \right) \\
&\overset{(a)}{=} \eta (F(p_{X}^{(0)})) = H(p_{X}^{(0)}) \approx 0.3251, \label{eq:approx_01}
\end{align}
where $(a)$ follows from Eq.~\eqref{eq:identity}. 
By the definition of $H_{2}^{\text{C}}(X|Y):= \min_{r_{X\mid Y}}-2\sum_{x}p_{X}(x)\log\sum_{y}p_{Y\mid X}(y | x)r_{X\mid Y}(x | y)^{\frac{1}{2}}$, 
choosing $(r_{X\mid Y}(0|0), r_{X\mid Y}(1|0))=(1, 0)$ and $(r_{X\mid Y}(0|1), r_{X\mid Y}(1|1))=(0, 1)$ yields 
\begin{align}
H_{2}^{\text{C}}(X|Y) &\leq -2\log 0.9 \approx 0.2107. \label{eq:approx_02}
\end{align}
Therefore, Eqs.~\eqref{eq:approx_01} and \eqref{eq:approx_02} yield a contradiction.
\end{proof}

\section{Proof of Theorem \ref{thm:CRE_DPI_g-Bayes}} \label{proof:CRE_DPI_g-Bayes}

\begin{proof}
$1)$ It suffices to show that $V_{\varphi, g}(X)\leq V_{\varphi, \psi, g}(X|Y)$. 
For an arbitrary $a\in \mathcal{A}$, consider the constant decision rule $\tilde{\delta}(y)\equiv a$. 
Then, we obtain the following: 
\begin{align}
V_{\varphi, \psi, g}(X|Y) &\geq \mathbb{M}_{\varphi}\left[\mathbb{M}_{\psi}\left[ g(X, \tilde{\delta}(Y)) \mid X\right]\right] \\ 
&= \mathbb{M}_{\varphi}[g(X, a)].
\end{align}
Because $a\in \mathcal{A}$ is arbitrary, it follows that 
\begin{align}
V_{\varphi, \psi, g}(X|Y) &\geq \max_{a}\mathbb{M}_{\varphi}[g(X, a)] = V_{\varphi, g}(X). 
\end{align}

$2)$ Suppose that $X-Y-Z$ forms a Markov chain. 
It suffices to show that $V_{\varphi, \psi, g}(X | Z)\leq V_{\varphi, \psi, g}(X|Y)$. 
For an arbitrary decision rule $\delta\colon \mathcal{Z}\to \mathcal{A}$, we obtain 
\begin{align}
&\mathbb{M}_{\psi}[g(X, \delta(Z)) \mid X=x] \notag \\ 
&= \psi^{-1}\left( \sum_{z}p_{Z\mid X}(z\mid x)\psi(g(x, \delta(z))) \right) \\ 
&= \psi^{-1}\left( \sum_{z}\left( \sum_{y}p_{Z\mid X, Y}(z | x, y)p_{Y\mid X}(y | x) \right)  \psi(g(x, \delta(z)))\right) \\
&\overset{(a)}{=} \psi^{-1}\left( \sum_{y}p_{Y\mid X}(y\mid x)\sum_{z}p_{Z\mid Y}(z\mid y)\psi(g(x, \delta(z))) \right) \\ 
&\overset{(b)}{\leq} \psi^{-1}\left( \sum_{y}p_{Y\mid X}(y\mid x)\psi(g(x, \vE_{Z}[\delta(Z)\mid Y=y])) \right) \\ 
&= \mathbb{M}_{\psi}[g(X, \vE_{Z}[\delta(Z)\mid Y]) \mid X=x], 
\end{align}
where $(a)$ follows from the Markov property $X-Y-Z$, and $(b)$ follows from Jensen's inequality. 
Thus, by the monotonicity of the KN mean $\mathbb{M}_{\varphi}[\cdot]$, we obtain 
\begin{align}
\mathbb{M}_{\varphi}\left[\mathbb{M}_{\psi}[g(X, \delta(Z)) \mid X]\right]
&\leq \mathbb{M}_{\varphi}\left[\mathbb{M}_{\psi}[g(X, \vE_{Z}[\delta(Z) | Y]) \mid X]\right] \\
&\leq V_{\varphi, \psi, g}(X|Y).
\end{align}
Because $\delta$ is arbitrary, we obtain 
\begin{align}
V_{\varphi, \psi, g}(X|Z) &\leq V_{\varphi, \psi, g}(X|Y).
\end{align}
This proves \texttt{DPI}. 
\end{proof}


\end{document}